\documentclass[conference]{IEEEtran}
\IEEEoverridecommandlockouts
\usepackage{cite}
\usepackage{amsmath,amssymb,amsfonts}
\usepackage{graphicx}
\usepackage{textcomp}
\usepackage{xcolor}
\usepackage{colortbl}
\usepackage{algorithm}
\usepackage{algpseudocode}
\usepackage{amsthm}

\newtheorem{theorem}{Theorem}[section]

\newtheorem{lemma}[theorem]{Lemma}

\newtheorem{definition}[theorem]{Definition}
\newtheorem{assumption}[theorem]{Assumption}
\newtheorem{remark}[theorem]{Remark}

\makeatletter
\let\NAT@parse\undefined
\makeatother
\usepackage{hyperref}       

\def\BibTeX{{\rm B\kern-.05em{\sc i\kern-.025em b}\kern-.08em
    T\kern-.1667em\lower.7ex\hbox{E}\kern-.125emX}}
\begin{document}

\title{Model-free Resilient Controller Design based on Incentive Feedback Stackelberg Game and Q-learning
}

\author{\IEEEauthorblockN{1\textsuperscript{st} Jiajun Shen}
\IEEEauthorblockA{\textit{The Department of Mechanical Engineering} \\
\textit{The University of Kansas}\\
Lawrence, KS 44906, USA \\
sjjvic@ku.edu}
\and
\IEEEauthorblockN{2\textsuperscript{nd} Fengjun Li}
\IEEEauthorblockA{\textit{The Department of Electrical Engineering and Computer Science} \\
\textit{The University of Kansas}\\
Lawrence, KS 44906, USA \\
fli@ku.edu}
\and
\IEEEauthorblockN{3\textsuperscript{rd} Morteza Hashemi}
\IEEEauthorblockA{\textit{The Department of Electrical Engineering and Computer Science} \\
\textit{The University of Kansas}\\
Lawrence, KS 44906, USA \\
mhashemi@ku.edu}
\and
\IEEEauthorblockN{4\textsuperscript{th} Huazhen Fang}
\IEEEauthorblockA{\textit{The Department of Mechanical Engineering} \\
\textit{The University of Kansas}\\
Lawrence, KS 44906, USA \\
fang@ku.edu}
}

\maketitle

\begin{abstract}
In the swift evolution of Cyber-Physical Systems (CPSs) within intelligent environments, especially in the industrial domain shaped by Industry 4.0, the surge in development brings forth unprecedented security challenges. This paper explores the intricate security issues of Industrial CPSs (ICPSs), with a specific focus on the unique threats presented by intelligent attackers capable of directly compromising the controller, thereby posing a direct risk to physical security. Within the framework of hierarchical control and incentive feedback Stackelberg game, we design a resilient leading controller (leader) that is adaptive to a compromised following controller (follower) such that the compromised follower acts cooperatively with the leader, aligning its strategies with the leader's objective to achieve a team-optimal solution. First, we provide sufficient conditions for the existence of an incentive Stackelberg solution when system dynamics are known. Then, we propose a Q-learning-based Approximate Dynamic Programming (ADP) approach, and corresponding algorithms for the online resolution of the incentive Stackelberg solution without requiring prior knowledge of system dynamics. Last but not least, we prove the convergence of our approach to the optimum.
\end{abstract}

\begin{IEEEkeywords}
Industrial Cyber-Physical Systems;  Incentive feedback Stackelberg Game; Resilient control; Q-learning; Approximate dynamic programming.
\end{IEEEkeywords}

\section{Introduction}

The exponential growth of smart and intelligent environments has catalyzed the rapid development of Cyber-Physical Systems (CPSs). Among the various applications of CPSs, industrial environments, including manufacturing \cite{dafflon2021challenges}, chemical production processes \cite{ji2016study}, and smart grids \cite{zhang2021smart}, stand out as the most important components of the fourth industrial revolution, known as Industry 4.0 \cite{dafflon2021challenges}.

\subsection{ICPS Security and Resilient Control}

Despite the evident advantages of integrating cyber and physical components for enhanced production efficiency, the security challenges associated with Industrial CPSs (ICPSs) have become increasingly intricate. Recent research has been focusing on securing the cyber layer, cyber-physical interactions, and the physical layer. Defending against cyber threats can be referred to as the traditional cybersecurity experiences, addressing issues like DNS hijacking, IP spoofing \cite{narula2022novel}, and SSH password attacks \cite{evers2017security}. Defending against cyber-physical threats can be referred to as addressing issues like information leakage and software update manipulation attacks (malicious files) against Gateway devices and web servers \cite{elhabashy2019cyber}.

Research on physical layer security primarily focuses on secure state estimation and resilient control. The main focus of secure state estimation is the design of estimators and filters \cite{kazemi2021finite, li2023efficient}.

Resilient control strategies aim to enable systems to recover from unforeseen adverse situations. Various studies have designed resilient controllers to tackle attacks ranging from jamming and Denial of Service (DoS) attacks to actuator and sensor attacks. Specifically, \cite{yuan2013resilient} designed a resilient controller for malicious jamming and DoS attacks on the communication channel. Based on the work of \cite{yuan2013resilient}, the resilient control under an intelligent DoS attacker (time-varying attack rates) is discussed in \cite{yuan2016resilient}. \cite{sun2019resilient} considered a DoS attacker targeting at blocking the controller-to-actuator (C-A) communication channel by launching adversarial jamming signals. Then, \cite{zhao2021adaptive} further considered both actuator attacks and sensor attacks and designed resilient controllers based on complex nonlinear system models caused by the unknown actuator and sensor attacks.

While existing research emphasizes resilience against communication channel delays, practical scenarios often involve intelligent attackers who are experts in reverse engineering. These attackers can compromise controllers by manipulating control codes, e.g., they can decode the running control policy/strategy, and they can also compromise the controller by injecting the malicious control codes \cite{cook2023survey}. This poses a unique challenge to the security of ICPS.
 
Consider a system with two controllers, such as the hierarchical control framework in ICPS deployment \cite{ghaeini2016hamids}. In this setup, a discrete control system (DCS) controller in the process control layer collaborates with a programmable logic controller (PLC) in the Fieldbus layer \footnote{Both process control layer and Fieldbus layer belong to the physical layer.}. The intelligent attacker can compromise PLC by injecting malicious but legitimate control codes to achieve arbitrary targets. This type of attack is particularly stealthy and harmful, posing challenges to detection mechanisms and defense strategies, for two reasons: 1) The target of intelligent attacks is mostly performance degradation in the long run, and thus the malicious control code is legitimate, and will not cause any operational abnormalcy; 2) Even when we can recognize the compromised PLC, it is still hard to mitigate its influence, since it is impractical to shut down production to refresh the control code considering the economic loss.

Therefore, regarding the physical security of ICPS, a resilient controller is required to not only bring the system back to desirable performance but also be adaptive to the attacker's different targets.

\subsection{Incentive Feedback Stackelberg Game}

The Stackelberg game, a pivotal tool for hierarchical decision problems, originated as a solution for static economic competition \cite{hicks1935marktform}. However, in most Stackelberg games, the leader may not face his most desirable outcome. In addressing this issue, so-called incentive mechanisms have been introduced to align the follower's optimum with the leader's desire.

In the context of ICPS, the DCS controller and PLC can be considered as the leader and follower in the Stackelberg game. When the PLC is uncompromised, the Stackelberg game simplifies into a joint optimization problem. However, in the event of PLC compromise, the problem transforms into \textit{designing an incentive strategy for the DCS controller to achieve its target despite the compromised PLC}.

The incentive Stackelberg game relies on the leader proposing a reward or penalty to the follower, altering the structure of the follower's optimization problem to induce a strategy aligned with the leader's desire, which is called a team-optimal solution. The author in \cite{ho1982control} suggested an incentive form $u_{\text{leader}} = u_{\text{leader}}^t + M (u_{\text{follower}} - u_{\text{follower}}^t)$ where superscript $t$ indicates the team-optimal solution, and $M$ is the appropriate incentive matrix. $(u_{\text{follower}} - u_{\text{follower}}^t)$ can be viewed as the ``penalty'' term for the follower for its deviation from the team-optimal solution. Based on this form, \cite{ahmed2017constrained, mukaidani2020robust} considered the state-feedback strategy, and investigated the incentive Stackelberg games with $H_{\infty}$ constraints, with one leader and multiple followers, and with Markovian jumps in dynamics, under both discrete-time and continuous-time settings. \cite{basar1979closed, li2002approach, lin2022incentive} studied different representations of incentive strategy by considering different forms of ``penalty'' term. They presented sufficient conditions for the incentive matrix $M$ under both deterministic and stochastic systems. However, all these researches considered the model-based and offline setting, i.e., they all require the knowledge of precise system dynamics. Besides, the matrix $M$ can only be derived by solving the complex matrix equations (e.g., cross-coupled Riccati equations), which is computationally inefficient.

Our study focuses on the incentive feedback Stackelberg game for discrete-time deterministic systems, employing a Q-learning-based approximate dynamic programming (ADP) approach. Unlike previous research, our contributions include 1) deriving a closed form for the incentive matrix, 2) developing a model-free (online) approach for team-optimal solutions and incentive matrix derivation, and 3) proving convergence to the optimum.

In the subsequent sections, we formally introduce the problem formulation in Section \ref{section: Problem Formulation}, solve the incentive feedback Stackelberg game with known dynamics in Section \ref{section: Incentive Feedback Stackelberg Game with Known Dynamics}, and present a model-free approach using Q-learning-based ADP. Section \ref{section: Q-learning-based ADP} is devoted to developing a model-free approach to derive a team-optimal solution and closed-form incentive matrix $M$ without the knowledge of system dynamics. Two corresponding algorithms and the proofs of the convergence to the optimum are given. Finally, we conclude by summarizing our contributions and outlining potential future directions in Section \ref{section: Conclusions}.


\textbf{Notations:} $\mathbb{E}(\cdot)$ is the mathematical expectation operator, $\mathbb{R}^{n}$ is the space of all real $n$-dimensional vectors, $\mathbb{R}^{m \times p}$ is the space of all $m \times p$ real matrices, $(\cdot)^T$ indicates the transpose operation, $M > 0$ and $M \geq 0$ indicates that matrix $M$ is positive definite and positive semi-definite, $\|\cdot\|_F$ indicates the Frobenius norm.

\section{Problem Formulation}
\label{section: Problem Formulation}

Consider the discrete-time systems governed by the following difference equation
\begin{equation}
    \label{eqn: dynamics}
    x_{k+1} = Ax_k + B_1u_k + B_2v_k,
\end{equation}
where $x_k \in X \subseteq \mathbb{R}^n$ is the system state, $u_k \in U \subseteq \mathbb{R}^{m_1}$ is controller $1$'s input, $v_k \in V \subseteq \mathbb{R}^{m_2}$ controller $2$'s input, $A$, $B_1$, and $B_2$ are matrices of appropriate dimensions. All this information and the value of the initial state, $x_0$ are known to both players.

\begin{assumption}
    \label{assumption: strategy information structure}
    All the controllers employ the closed-loop memoryless policies, i.e., $u_k = \pi_1(k, x_0, x_k)$, $v_k = \pi_2(k, x_0, x_k)$ \cite{bacsar1998dynamic}. In addition, the linear closed-loop memoryless Stackelberg strategy has the following form \cite{bacsar1998dynamic, medanic1978closed}:
    \begin{equation}
        \pi_i(k, x_0, x_k) = K_i x_k, \ i=1,2,
    \end{equation} where $K_1 \in \mathbb{R}^{m_1 \times n}$, $K_2 \in \mathbb{R}^{m_2 \times n}$ are matrices with appropriate dimensions.
\end{assumption}

Consider the state-feedback policy for controller $1$ (resp. controller $2$) $\pi_1 \in \Pi_1:\mathbb{R}^n \rightarrow \mathbb{R}^{m_1}$ (resp. $\pi_2 \in \Pi_2:\mathbb{R}^n \rightarrow \mathbb{R}^{m_2}$) where $\Pi_1$ and $\Pi_2$ are sets of admissible policies, and specifically are of state-feedback form, i.e., $u_k = \pi_1(x_k)=K_1 x_k$, and $v_k = \pi_2(x_k)=K_2 x_k$.

The infinite-horizon cost functions of controller $1$ and $2$ are given respectively by
\begin{equation}
    \label{eqn: cost function of leader}
    J_1(\pi_1,\pi_2) = \sum_{k=0}^{\infty} \gamma^k (x_k^T Q_1 x_k + u_k^T R_{11} u_k + v_k^T R_{12} v_k),
\end{equation}
\begin{equation}
    \label{eqn: cost function of follower}
    J_2(\pi_1,\pi_2) = \sum_{k=0}^{\infty} \gamma^k (x_k^T Q_2 x_k + u_k^T R_{21} u_k + v_k^T R_{22} v_k),
\end{equation}
where $Q_1 = Q_1^T \geq 0$, $Q_2 = Q_2^T \geq 0$, $R_{11} = R_{11}^T > 0$, $R_{12} = R_{12}^T > 0$, $R_{21} = R_{21}^T > 0$, and $R_{22} = R_{22}^T > 0$ are known coefficient matrices, $\gamma \in (0,1)$ is the discount factor.

Define $c_{i,k} := c_i(x_k, u_k, v_k) = x_k^T Q_i x_k + u_k^T R_{i1} u_k + v_k^T R_{i2} v_k$, as the one-step cost at $k$-th step for both controllers where $i=1,2$, and $c_i$ is a cost function.

Given the policies of controller $1$ and $2$, $\pi=\{\pi_1, \pi_2\}$, the state-value functions $V_1^{\pi}: \mathbb{R}^n \rightarrow \mathbb{R}$, $V_2^{\pi}: \mathbb{R}^n \rightarrow \mathbb{R}$, and the action-value functions $Q_1^{\pi}: \mathbb{R}^n \times \mathbb{R}^{m_1} \rightarrow \mathbb{R}$, $Q_2^{\pi}: \mathbb{R}^n \times \mathbb{R}^{m_2} \rightarrow \mathbb{R}$ are defined as
\begin{equation}
    \label{eqn: value function of leader}
    V_i^{\pi}(x_k)=\min_{\substack{u_k,u_{k+1},\cdots \\ v_k,v_{k+1},\cdots}} \sum_{j=0}^{\infty} \gamma^{k+j} c_{i,k+j},
\end{equation}
\begin{equation}
    \label{eqn: action-value function of leader}
    Q_i^{\pi}(x_k, u_k, v_k)= c_{i,k} + \min_{\substack{u_{k+1},u_{k+2},\cdots \\ v_{k+1},v_{k+2},\cdots}} \sum_{j=0}^{\infty} \gamma^{k+j} c_{i,k+j}.
\end{equation}

Consider the controllers as two players in the Stackelberg game setting. Without loss of generality, we assume controller $1$ as leader, and controller $2$ as follower. Then, the Stackelberg solution should satisfy:
\begin{equation}
    \label{eqn: definition of Stackelberg solution}
    J_1(\pi_1^*, \pi_2^*) := J_1(\pi_1^*, R_2(\pi_1^*)) = \min_{\pi_1} J_1(\pi_1, R_2(\pi_1)),
\end{equation}
where $R_2(\pi_1) = \{\pi \in \Pi_2: J_2(\pi_1, \pi) \leq J_2(\pi_1, \pi_2), \forall \pi_2 \in \Pi_2\}$ is the rational reaction set of the follower.

\begin{assumption}
    \label{assumption: information structure}
    The leader has access to the follower's strategy, i.e., the leader has access to policy $\pi_2$.
\end{assumption}

\begin{definition}
    \label{def: team-optimal solution}
    A strategy pair $(\pi_1^t, \pi_2^t)$ is called the team-optimal solution of the game if \begin{equation}
        \label{eqn: team-optimal solution inequality}
        J_1(\pi_1^t, \pi_2^t) \leq J_1(\pi_1, \pi_2), \ \forall \pi_1 \in \Pi_1 \ \text{and} \ \forall \pi_2 \in \Pi_2.
    \end{equation}
\end{definition}

\begin{remark}
    The team-optimal solution $(\pi_1^t, \pi_2^t)$ can only be achieved when both players act ``cooperatively''. In other words, the follower would help the leader achieve the leader's desired target, i.e., minimizing $J_1$, while achieving his own desired target, i.e., minimizing $J_2$. Consider a special and the ideal case (for leader) when $Q_1 = Q_2$, $R_{11} = R_{21}$, and $R_{12} = R_{22}$, and thus the targets of leader and follower collapse to the same one. In this case, the team-optimal solution is guaranteed to be consistent with the incentive Stackelberg solution. However, in most Stackelberg games, the team-optimal solution is hard to achieve due to the follower's different desired target, which would result in the gap between $\pi_2^*$ and $\pi_2^t$.
\end{remark}

In this paper, we adopt a similar incentive form as suggested in \cite{ho1982control, ahmed2017constrained, mukaidani2018incentive, mukaidani2020robust}, $u_k = u_k^t + M(v_k - v_k^t)$ where $M$ is the incentive matrix to be determined, and superscript $t$ represents the team-optimal value. The second term on the right-hand side (RHS) can be viewed as a punishment for the follower's deviation from the team-optimal solution, $v_k^t$.

\section{Incentive Feedback Stackelberg Game with Known System Dynamics}
\label{section: Incentive Feedback Stackelberg Game with Known Dynamics}

In this section, we introduce how to design an incentive feedback Stackelberg strategy that achieves the team optimum, given known dynamics. We first derive the team-optimal solution by Lemma \ref{lemma: team-optimal}.

\begin{lemma}
    \label{lemma: team-optimal}
    Given Assumption \ref{assumption: strategy information structure} is satisfied, the joint optimization problem
    \begin{equation}
        \label{eqn: joint optimization}
        \begin{aligned}
           &\min J_1(\pi_1,\pi_2) \\
           &= \sum_{k=0}^{\infty} \gamma^k (x_k^T Q_1 x_k + u_k^T R_{11} u_k + v_k^T R_{12} v_k), \\
           & u_k = \pi_1(x_k), v_k = \pi_2(x_k),\\ 
           & \text{s.t.} \ x_{k+1} = Ax_k + B_1u_k + B_2v_k,
        \end{aligned}
    \end{equation} admits a unique team-optimal solution $\{\pi_1^t, \pi_2^t\}$
    \begin{equation}
        \label{eqn: leader team-optimal solution}
        u_k^t = \pi_1^t (x_k) = -K_1 x_k, 
    \end{equation}
    \begin{equation}
        \label{eqn: follower team-optimal solution}
        v_k^t = \pi_2^t (x_k) = -K_2 x_k,
    \end{equation} and with minimum cost $J_1^t = x_0^T P x_0$, where
    \begin{equation}
        \label{eqn: K_i}
        K_i = \gamma{(R_{1i} + \gamma F_i B_i)}^{-1} F_i A,
    \end{equation}
    \begin{equation}
        \label{eqn: F_i}
        \begin{aligned}
             F_i = B_i^T P \big[I - \gamma B_j{[R_{1j} + \gamma B_j^T P B_j]}^{-1} B_j^T P\big],& \\
             \ i,j=1,2, \ i \neq j&
        \end{aligned}
    \end{equation}
    \begin{equation}
        \label{eqn: ARE of team-optimal}
        \begin{aligned}
            P &= Q_1 + \gamma {(A - B_1 K_1-B_2 K_2)}^T P (A - B_1 K_1 - B_2 K_2) \\
            & + K_1^T R_{11} K_1 + K_2^T R_{12} K_2.
        \end{aligned}
    \end{equation}
\end{lemma}

\begin{proof}
    Since the state value function is quadratic and policies are state-feedback, we have $V_1^{\pi^t}(x_t) = x_t^T P x_t$, where $P = P^T \geq 0$, and $\pi^t = \{\pi_1^t, \pi_2^t\}$ is the team-optimal strategy. Also, we have the Bellman equation
    \begin{equation}
        \begin{aligned}
            &V_1^{\pi^t}(x_k) = \min_{u_k, v_k} \big(x_k^T Q_1 x_k + u_k^T R_{11} u_k + v_k^T R_{12} v_k \\
            &+ \gamma V_1^{\pi^t}(Ax_k + B_1u_k + B_2v_k)\big)
        \end{aligned}
    \end{equation}

    We begin with the derivation of the person-by-person (PBP) optimal solution of the joint optimization problem \eqref{eqn: joint optimization}. 
    
    First, given \eqref{eqn: follower team-optimal solution}, we have the following standard optimal control problem
    \begin{equation}
        \label{eqn: leader optimal control problem}
        \begin{aligned}
           &\min J_1(\pi_1) = \sum_{k=0}^{\infty} \gamma^k x_k^T [Q_1 + K_2^T R_{12} K_2] x_k + u_k^T R_{11} u_k, \\
           & u_k = \pi_1(x_k), \\
           & \text{s.t.} \ x_{k+1} = (A - B_2K_2)x_k + B_1u_k.
        \end{aligned}
    \end{equation}

    The optimal state-feedback strategy of the leader is given by
    \begin{equation}
        \label{eqn: K_1}
        K_1 = \gamma {(R_{11} + \gamma B_1^T P B_1)}^{-1} B_1^T P (A - B_2K_2).
    \end{equation}
    Similarly, we can derive the following optimal state-feedback strategy of the follower given \eqref{eqn: leader team-optimal solution}
    \begin{equation}
        \label{eqn: K_2}
        K_2 = \gamma {(R_{12} + \gamma B_2^T P B_2)}^{-1} B_2^T P (A - B_1K_1).
    \end{equation}

    By substituting \eqref{eqn: K_2} into \eqref{eqn: K_1} and doing some calculation, we have
    \begin{equation}
        (R_{11} + \gamma F_1B_1) K_1 = \gamma F_1 A,
    \end{equation}
    which gives us the case of $i=1, j=2$ in \eqref{eqn: K_i}. Similarly, substituting \eqref{eqn: K_1} into \eqref{eqn: K_2} would give us the case of $i=2, j=1$.

    For the standard optimal control problem \eqref{eqn: leader optimal control problem}, we have the following algebraic Riccati equation (ARE)
    \begin{equation}
        \label{eqn: ARE of leader optimal control}
        \begin{aligned}
            &\gamma {(A - B_2K_2)}^T P (A - B_2K_2) - P + Q_1 + K_2^T R_{12} K_2 \\
            & - \gamma {(A - B_2K_2)}^T P B_1 K_1 = 0,
        \end{aligned}
    \end{equation}
    where 
    \begin{equation}
        \begin{aligned}
            & \gamma {(A - B_2K_2)}^T P (A - B_2K_2) - \gamma {(A - B_2K_2)}^T P B_1 K_1\\
            &= \gamma {(A - B_1K_1 - B_2K_2)}^T P (A - B_1K_1 - B_2K_2) \\
            & + \gamma K_1^T B_1^T P (A- B_2K_2) - \gamma K_1^T B_1^T B_1 K_1\\
            & = \gamma {(A - B_1K_1 - B_2K_2)}^T P (A - B_1K_1 - B_2K_2) \\
            &+ K_1^T R_{11} K_1,
        \end{aligned}
    \end{equation} which leads to \eqref{eqn: ARE of team-optimal}.
\end{proof}

Then, the leader is supposed to announce the incentive strategy, $u_k = u_k^t + M(v_k - v_k^t)$, in advance to the follower. Accordingly, the follower needs to solve for his own optimal strategy.

\begin{lemma}
    \label{lemma: follower optimal strategy}
    Given the leader's incentive strategy $u_k = u_k^t + M(v_k - v_k^t)$, the follower's optimization problem is defined as
    \begin{equation}
        \label{eqn: follower optimal control problem under incentive}
        \begin{aligned}
       &\min J_2(\pi_2) = \sum_{k=0}^{\infty} \gamma^k x_k^T Q_M x_k + 2 x_k^T R_{1,M} v_k + v_k^T R_{2,M} v_k, \\
       & v_k = \pi_2(x_k), \\
       & \text{s.t.} \ x_{k+1} = A_M x_k + B_M v_k,
        \end{aligned}
    \end{equation} where $Q_M:=Q_2 + K_1^T R_{21} K_1 - 2K_1^T R_{21} M K_2 + K_2^T M^T R_{21} M K_2$, $R_{1,M} := -2K_1^T R_{21} M + 2K_2^T M^T R_{21} M$ , $R_{2,M} := R_{22} + M^T R_{21} M$, $A_M := A - B_1K_1 + B_1MK_2$, $B_M := B_1M +B_2$, and $K_1$, $K_2$ satisfy \eqref{eqn: K_i}, \eqref{eqn: F_i}, and \eqref{eqn: ARE of team-optimal}.

    It admits a unique optimal solution
    \begin{equation}
        \label{eqn: follower optimal solution}
        v_k^* = - {(R_{2,M} + \gamma B_M^T P_v B_M)}^{-1} {(R_{1,M} + \gamma A_M^T P B_M)}^T x_k
    \end{equation} where $P_v$ satisfies the following ARE
    \begin{equation}
        \label{eqn: ARE M version}
        \begin{aligned}
            P_v =& Q_M + \gamma A_M^T P_v A_M - [R_{1,M} + \gamma A_M^T P_v B_M] \\
            &\cdot {[R_{2,M} + \gamma B_M^T P_v B_M]}^{-1} {[R_{1,M} + \gamma A_M^T P_v B_M]}^T.
        \end{aligned}
    \end{equation}
\end{lemma}

\begin{proof}
    The proof follows the standard linear quadratic discrete-time regulator problem.
\end{proof}



Now, we are ready to provide the main result of the incentive Stackelberg strategy for the leader such that the team-optimal solution can be achieved.

\begin{theorem}
    Consider the Stackelberg game captured by dynamics \eqref{eqn: dynamics} and cost functions of leader and follower, \eqref{eqn: cost function of leader} and \eqref{eqn: cost function of follower}, the team-optimal solution $(\pi_1^t, \pi_2^t)$ defined by \eqref{eqn: leader team-optimal solution}, and \eqref{eqn: follower team-optimal solution}, can be achieved if the leader chooses the incentive strategy $u_k = u_k^t + M(v_k - v_k^t)$ where
    \begin{equation}
        \label{eqn: M}
        \begin{aligned}
            M =& {\big(\gamma {(A - B_1K_1 - B_2K_2)}^T P_v B_1 - K_1^TR_{21}\big)}^{-1} \\
            &\cdot\big(K_2^T R_{22} - \gamma {[A - B_1K_1 - B_2K_2]}^T P_v B_2\big),
        \end{aligned}
    \end{equation} where $K_1$, $K_2$ satisfy \eqref{eqn: K_i}, \eqref{eqn: F_i}, \eqref{eqn: ARE of team-optimal}, and $P_v$ satisfies the following ARE
    \begin{equation}
        \label{eqn: P_v}
            \begin{aligned}
                P_v &= Q_2 + \gamma {(A - B_1 K_1-B_2 K_2)}^T P_v (A - B_1 K_1 - B_2 K_2) \\
                & + K_1^T R_{21} K_1 + K_2^T R_{22} K_2.
            \end{aligned}
    \end{equation}
\end{theorem}

\begin{proof}
    By equating the follower's optimal solution $v_k^*$ \eqref{eqn: follower optimal solution} with the team-optimal solution $v_k^t$ \eqref{eqn: follower team-optimal solution}, i.e., $K_2 = - {(R_{2,M} + \gamma B_M^T P_v B_M)}^{-1} {(R_{1,M} + \gamma A_M^T P B_M)}^T x_k$, and substituting the $A_M$, $B_M$ and $R_{1,M}$ as defined in Lemma \ref{lemma: follower optimal strategy}, \eqref{eqn: ARE M version} becomes
    \begin{equation}
        \label{eqn: sub1 follower ARE}
        \begin{aligned}
            &P_v - Q_2 - K_1^T R_{21} K_1 + K_1^T R_{21} M K_2 - \gamma {(A - B_1K_1)}^T \\
            & \cdot P_v [A - B_1K_1 - B_2K_2] \\
            &= \gamma {(B_1 M K_2)}^T P_v [A - B_1K_1 - B_2K_2].
        \end{aligned}
    \end{equation}

    Then, substituting the $A_M$, $B_M$ and $R_{1,M}$ into $K_2 = - {(R_{2,M} + \gamma B_M^T P_v B_M)}^{-1} {(R_{1,M} + \gamma A_M^T P B_M)}^T x_k$, we have
    \begin{equation}
        \label{eqn: M derivation equation}
        \begin{aligned}
            &R_{22} K_2 + \gamma {(B_1 M + B_2)}^T P_v [-A + B_1K_1 + B_2K_2] \\
            &= - M^T R_{21} K_1.
        \end{aligned}
    \end{equation}
    Solving \eqref{eqn: M derivation equation} for the closed-form of $M$  leads us to \eqref{eqn: M}. Then, by multiplying both sides of \eqref{eqn: M derivation equation} by $K_2^T$, we have
    \begin{equation}
        \label{eqn: sub2 follower ARE}
        \begin{aligned}
            &K_2^T M^T R_{21} K_1 + K_2^T R_{22} K_2 - \gamma K_2^T B_2^T P_v \\
            &\cdot [A - B_1K_1 - B_2K_2] \\
            &= \gamma {(B_1 M K_2)}^T P_v [A - B_1K_1 - B_2K_2].
        \end{aligned}
    \end{equation}
    Equating \eqref{eqn: sub1 follower ARE} with \eqref{eqn: sub2 follower ARE} gives us \eqref{eqn: P_v}.
\end{proof}

\section{Q-learning-based Approximate Dynamic Programming with Unknown Dynamics}
\label{section: Q-learning-based ADP}

In this section, a Q-learning-based approximate dynamic programming (ADP) approach is developed that solves the incentive Stackelberg solution for the leader online without requiring any knowledge of the system dynamics $(A, B_1, B_2)$.

\subsection{Q-function for joint optimization problem}

The optimal action-value function $Q_1^{\pi^t}$ (associated with the team-optimal solution $\pi^t=\{\pi_1^t, \pi_2^t\}$) is defined as
\begin{equation}
    \label{eqn: optimal H}
    \begin{aligned}
        & Q_1^{\pi^t}(x_k, u_k, v_k) = c_1(x_k, u_k, v_k) + \gamma V_1^{\pi^t} (x_{k+1}) \\
        & = [x_k^T \ u_k^T \ v_k^T] H {[x_k^T \ u_k^T \ v_k^T]}^T,
    \end{aligned}
\end{equation} where $H \in \mathbb{R}^{l \times l}$, $l=n+m_1+m_2$, associated with $P$ that solves \eqref{eqn: ARE of team-optimal}.

The relationship between $H$ and $P$ can be derived as
\begin{equation}
    \label{eqn: H and P equation}
    \begin{aligned}
        & [x_k^T \ u_k^T \ v_k^T] H {[x_k^T \ u_k^T \ v_k^T]}^T \\
        & = x_k^T Q_1 x_k + u_k^T R_{11} u_k + v_k^T R_{12} v_k + x_{k+1}^T P x_{k+1} \\
        & = [x_k^T \ u_k^T \ v_k^T] \textbf{diag}(Q_1, R_{11}, R_{12}) {[x_k^T \ u_k^T \ v_k^T]}^T \\
        &+ \gamma [x_k^T \ u_k^T \ v_k^T] {[A \ B_1 \ B_2]}^T P [A \ B_1 \ B_2] {[x_k^T \ u_k^T \ v_k^T]}^T.
    \end{aligned}
\end{equation}

$H$ can be written in block matrix form as
\begin{equation}
    \label{eqn: H}
    \begin{aligned}
        &\begin{bmatrix}
        H_{xx} & H_{xu} & H_{xv}\\
        H_{ux} & H_{uu} & H_{uv}\\
        H_{vx} & H_{vu} & H_{vv}
        \end{bmatrix}\\
        & = \begin{bmatrix}
        Q_1 + \gamma A^T P A & \gamma A^T P B_1 & \gamma A^T P B_2\\
        \gamma B_1^T P A & R_{11} + \gamma B_1^T P B_1 & \gamma B_1^T P B_2\\
        \gamma B_2^T P A & \gamma B_2^T P B_1 & R_{12} + \gamma B_2^T P B_2
        \end{bmatrix},
    \end{aligned}
\end{equation}

Note that, for any $x_k \in X$, we have
\begin{equation}
    \begin{aligned}
        V_1^{\pi^t} (x_k) = \min_{u_k, v_k} Q_1^{\pi^t}(x_k, u_k, v_k) = Q_1^{\pi^t} (x_k, u_k^t, v_k^t),
    \end{aligned}
\end{equation} where $u_k^t$ and $v_k^t$ are team-optimal solution as defined in \eqref{eqn: leader team-optimal solution} and \eqref{eqn: follower team-optimal solution}.

By equating $Q_1^{\pi^t} (x_k, u_k^t, v_k^t)$ and $V_1^{\pi^t}(x_k)$, we have
\begin{equation}
    \label{eqn: H and P relation}
    P = H + K_1^T H K_1 + K_2^T H K_2 = [I \ K_1^T \ K_2^T] H {[I \ K_1^T \ K_2^T]}^T,
\end{equation}

By substituting \eqref{eqn: H and P relation} into \eqref{eqn: H and P equation}, we can derive the action-value function version of ARE and Bellman equation as follows
\begin{equation}
    \label{eqn: Q version ARE}
    \begin{aligned}
       &H = \textbf{diag}(Q_1, R_{11}, R_{12}) \\
       &+ \gamma \begin{bmatrix}
        A & B_1 & B_2\\
        K_1 A & K_1 B_1 & K_1 B_2\\
        K_2 A & K_2 B_1 & K_2 B_2
        \end{bmatrix}^T H \begin{bmatrix}
        A & B_1 & B_2\\
        K_1 A & K_1 B_1 & K_1 B_2\\
        K_2 A & K_2 B_1 & K_2 B_2
        \end{bmatrix}
    \end{aligned} 
\end{equation}
\begin{equation}
    \label{eqn: Q version Bellman equation}
    Q_1^{\pi^t}(x_k, u_k, v_k) = c_1(x_k, u_k, v_k) + Q_1^{\pi^t}(x_{k+1}, u_{k+1}^t, v_{k+1}^t),
\end{equation} where $u_{k+1}^t = -K_1 x_{k+1}$, and $v_{k+1}^t = -K_2 x_{k+1}$.

Using \eqref{eqn: H}, we can rewrite $K_1$ and $K_2$ as
\begin{equation}
    \label{eqn: rewrite K_1 using H}
    \begin{aligned}
        K_1 &= {(H_{uu} - H_{uv}{(H_{vv})}^{-1}H_{vu})}^{-1}\\
        & \cdot (H_{ux} - H_{uv}{(H_{vv})}^{-1}H_{vx}),
    \end{aligned}
\end{equation}
\begin{equation}
    \label{eqn: rewrite K_2 using H}
    \begin{aligned}
        K_2 &= {(H_{vv} - H_{vu}{(H_{uu})}^{-1}H_{uv})}^{-1}\\
        & \cdot (H_{vx} - H_{vu}{(H_{uu})}^{-1}H_{ux}).
    \end{aligned}
\end{equation}

From \eqref{eqn: rewrite K_1 using H} and \eqref{eqn: rewrite K_2 using H}, we observe that the team-optimal solution only depends on matrix $H$. Similar to $P$, $H$ can be derived by solving the corresponding ARE, which requires the knowledge of system dynamics $(A, B_1, B_2)$. However, if $H$ is known to us, we can derive the team-optimal solution without the knowledge of system dynamics $(A, B_1, B_2)$. Inspired by this observation, we are aiming to develop an approach to solve for $H$ with unknown dynamics.

\subsection{Online derivation of team-optimal solution}

In the traditional Q-learning setting, the agent updates the Q function according to the reward signal and the estimate of optimal future value (based on current Q function). Since each Q function is associated with a certain policy, the update of Q function implies the improvement of policy. This is under the policy iteration framework. Then, we define the updating rule of Q function (or equivalently policy) as
\begin{equation}
    \begin{aligned}
        &Q_1^{\pi_{i+1}} (x_k, u_k, v_k) = [x_{k}^T \ u_{k}^T \ v_{k}^T] H_{i+1} {[x_{k}^T \ u_{k}^T \ v_{k}^T]}^T\\
        &= x_k^T Q_1 x_k + u_k^T R_{11} u_k + v_k^T R_{12} v_k \\
        &\ \ \ \ + \min_{u_{k+1}, v_{k+1}} Q_1^{\pi_{1,i}}(x_{k+1}, u_{k+1}, v_{k+1}) \\
        & = x_k^T Q_1 x_k + u_k^T R_{11} u_k + v_k^T R_{12} v_k \\
        &\ \ \ \ + Q_1^{\pi_{1,i}}(x_{k+1}, u_{k+1,i}^t, v_{k+1,i}^t)\\
        &= x_k^T Q_1 x_k + u_k^T R_{11} u_k + v_k^T R_{12} v_k \\
        &\ \ \ \ + [x_{k+1}^T \ u_{k+1}^T \ v_{k+1}^T] H_i {[x_{k+1}^T \ u_{k+1}^T \ v_{k+1}^T]}^T,
    \end{aligned}
\end{equation} where $i$ indicates the number of policy iteration, $\pi_{i+1}=\{\pi_{1,i+1}, \pi_{2,i+1}\}$, $u_{k+1,i}^t = \pi_{1,i}^t (x_{k+1}) = -K_{1,i} x_{k+1}$, $v_{k+1,i}^t = \pi_{2,i}^t (x_{k+1}) = -K_{2,i} x_{k+1}$, $K_{1,i}$ and $K_{2,i}$ are defined as follows
\begin{equation}
    \label{eqn: rewrite K_{1,i} using H}
    \begin{aligned}
        K_{1,i} &= {(H_{uu}^i - H_{uv}^i{(H_{vv}^i)}^{-1}H_{vu}^i)}^{-1}\\
        & \cdot (H_{ux}^i - H_{uv}^i{(H_{vv}^i)}^{-1}H_{vx}^i),
    \end{aligned}
\end{equation}
\begin{equation}
    \label{eqn: rewrite K_{2,i} using H}
    \begin{aligned}
        K_{2,i} &= {(H_{vv}^i - H_{vu}^i{(H_{uu}^i)}^{-1}H_{uv}^i)}^{-1}\\
        & \cdot (H_{vx}^i - H_{vu}^i{(H_{uu}^i)}^{-1}H_{ux}^i).
    \end{aligned}
\end{equation}

In order to solve the optimal Q-function (equivalently the optimal $H$) forward in time, we derive the following recurrence equation on $i$
\begin{equation}
    \label{eqn: recurrence on i}
    \begin{aligned}
        &Q_1^{\pi_{i+1}} (x_k, u_{k,i}^t, v_{k,i}^t) \\
        &= x_k^T Q_1 x_k + {(u_{k,i}^t)}^T R_{11} u_{k,i}^t + {(v_{k,i}^t)}^T R_{12} v_{k,i}^t \\
        & + [x_{k+1}^T \ {(u_{k+1,i}^t)}^T \ {(v_{k+1,i}^t)}^T] H_i \\
        &\cdot {[x_{k+1}^T \ {(u_{k+1,i}^t)}^T \ {(v_{k+1,i}^t)}^T]}^T,
    \end{aligned}
\end{equation} where $u_{k,i}^t = \pi_{1,i}^t (x_{k}) = -K_{1,i} x_{k}$, and $v_{k,i}^t = \pi_{2,i}^t (x_{k}) = -K_{2,i} x_{k}$.

Our goal is to prove that $Q_1^{\pi_{i}} \rightarrow Q_1^{\pi^t}$ as $i \rightarrow \infty$ which implies $\pi_{i} \rightarrow \pi^t$, $H_i \rightarrow H$, $K_{1,i} \rightarrow K_1$, and $K_{2,i} \rightarrow K_2$ as $i \rightarrow \infty$.

Then, in order to directly estimate the Q function, we rewrite the Q function in a parametric structure (parameterized by $H$) as
\begin{equation}
    \label{eqn: parameterized by H}
    Q_1^{\pi_i}(x_k, u_k, v_k) = z_k^T H_i z_k = {\bar{z}_k}^T \Theta(H_i),
\end{equation} where $z_k = {[x_k^T \ u_k^T \ v_k^T]}^T \in \mathbb{R}^l$, $\bar{z}_k \in \mathbb{R}^{l(l+1)/2}$ is the vector whose elements are all of the quadratic basis functions over the elements of $z_k$ (Kronecker product quadratic polynomial basis vector \cite{brewer1978kronecker}), i.e., $\bar{z}_k = (z_{k,1}^2, z_{k,1}z_{k,2}, \ldots, z_{k,1}z_{k,l}, z_{k,2}^2, z_{k,2}z_{k,3}, \ldots, z_{k,2}z_{k,l}, \ldots,$ 
$ z_{k,l-1}^2, z_{k,l-1}z_{k,l}, z_{k,l}^2)$. $\Theta(H_i) \in \mathbb{R}^{l(l+1)/2}$ is the vector whose elements are the $l$ diagonal entries of $H_i$ and the $(l(l+1)/2 - l)$ distinct sums of off-diagonal elements, $H_i[j,k] + H_i[k,j]$. $H_i[j,k]$ indicates the element of $H_i$ located at $j$-th row and $k$-th column. The original matrix $H_i$ can be retrieved from $\Theta(H_i)$ since $H_i$ is symmetric.

According to \eqref{eqn: parameterized by H}, $Q_1^{\pi_{i+1}} (x_k, u_{k,i}^t, v_{k,i}^t)$ is linearly parameterized by vector $\Theta(H_{i+1})$. Given that $H_i$ is known to us, we can view \eqref{eqn: recurrence on i} as the desired target function of the estimate of $Q_1^{\pi_{i+1}} (x_k, u_{k, i}^t, v_{k, i}^t)$, i.e., $\hat{Q}_1^{\pi_{i+1}} (x_k, u_{k, i}^t, v_{k, i}^t):= \bar{z}_k^T \hat{\Theta}(H_{i+1})$. Note that what we retrieve from the vector $\hat{\Theta}(H_{i+1})$ is the estimate of $H_{i+1}$, i.e., $\hat{H}_{i+1}$. 

Specifically, we consider the least-square approximation, i.e., find the parameter vector to minimize the error between the target value and estimate in a least-square sense over a compact set 
$X_c \subset X$,
\begin{equation}
    \label{eqn: least-square problem}
    \begin{aligned}
        &\hat{\Theta}(H_{i+1}) = \hat{h}_{i+1}\\
        &:= {\mathop{\arg\min}_{h}}\big(\int_{X_c}{\big|\bar{z}_k^T h - Q_1^{\pi_{i+1}} (x_k, u_{k,i}^t, v_{k,i}^t)\big|}^2\text{d}x_k\big).
    \end{aligned}
\end{equation}

Solving the least-square problem \eqref{eqn: least-square problem} gives us
\begin{equation}
    \hat{h}_{i+1} = {\bigg(\int_{X_c} \bar{z}_k {\bar{z}_k}^T \bigg)}^{-1} \int_{X_c} \bar{z}_k Q_1^{\pi_{i+1}} (x_k, u_{k,i}^t, v_{k,i}^t) \text{d}x.
\end{equation}

Note that $\bar{z}_k$ is the function of $x_k$, i.e., $\bar{z}_k(x_k)$ since $z_k = {[x_k^T \ {(u_{k,i}^t)}^T \ {(v_{k,i}^t)}^T]}^T$ where both $u_{k,i}^t$ and $v_{k,i}^t$ are linearly dependent on $x_k$. Thus, $\int_{X_c} \bar{z}_k {\bar{z}_k}^T \text{d}x$ is convertible, which implies that the least-square problem \eqref{eqn: least-square problem} is not well-defined. We introduce the exploration noise to both controller inputs to solve this issue, i.e.,
\begin{equation}
    \label{eqn: u+exploration}
    \hat{u}_{k,i}^t = u_{k,i}^t + \epsilon_{1,k} = - K_{1,i}x_k + \epsilon_{1,k},
\end{equation}
\begin{equation}
    \label{eqn: v+exploration}
    \hat{v}_{k,i}^t = v_{k,i}^t + \epsilon_{2,k} = - K_{2,i}x_k + \epsilon_{2,k},
\end{equation} where $\epsilon_{1,k} \sim N(0,\sigma_1)$ and $\epsilon_{2,k} \sim N(0,\sigma_2)$.

Then, the desired target defined by \eqref{eqn: recurrence on i} becomes
\begin{equation}
    \label{eqn: desired target with noise}
    \begin{aligned}
        &\hat{Q}_1^{\pi_{i+1}} (x_k, u_{k,i}^t, v_{k,i}^t) \\
        &= x_k^T Q_1 x_k + {(\hat{u}_{k,i}^t)}^T R_{11} \hat{u}_{k,i}^t + {(\hat{v}_{k,i}^t)}^T R_{12} \hat{v}_{k,i}^t \\
        & + [x_{k+1}^T \ {(u_{k+1,i}^t)}^T \ {(v_{k+1,i}^t)}^T] H_i \\
        &\cdot {[x_{k+1}^T \ {(u_{k+1,i}^t)}^T \ {(v_{k+1,i}^t)}^T]}^T \\
        &=\hat{Q}_1^{\pi_{i+1}} (x_k, H_i).
    \end{aligned}
\end{equation}

Given a sufficiently large set $X_c$, i.e., enough data points ($d_1, d_2, d_3, \ldots, d_N \in X_c$) collected, for solving the least-square problem \eqref{eqn: least-square problem}, we have
\begin{equation}
    \label{eqn: solution of least-square problem}
    \hat{h}_{i+1} = {\bigg(\hat{Z} {\big(\hat{Z}\big)}^T\bigg)}^{-1} \hat{Z} \hat{Q},
\end{equation} where $\hat{Z} = [\hat{z}(d_1), \hat{z}(d_2), \ldots, \hat{z}(d_N)]$, $\hat{z}(d_j) = {[d_j^T \ {(-K_{1,i} d_j + \epsilon_{1,k})}^T \ {(-K_{2,i} d_j + \epsilon_{2,k})}^T]}^T$, and $\hat{Q}={[\hat{Q}_1^{\pi_{i+1}} (d_1, H_i), \hat{Q}_1^{\pi_{i+1}} (d_2, H_i), \ldots, \hat{Q}_1^{\pi_{i+1}} (d_N, H_i)]}^T$.

The least-square problem \eqref{eqn: least-square problem} can be solved in an online fashion (i.e., without requiring any knowledge of system dynamics $(A, B_1, B_2)$), and under a policy iteration framework. It should be noted that, before implementing the policy iteration, we need to collect enough data tuples ${\{x_k, x_{k+1}\}}_{k=1,2,\ldots, N-1}$. In addition, since $H_i \in \mathbb{R}^{l \times l}$ is symmetric with $l(l+1)/2$ independent elements, at least $l(l+1)/2$ data tuples are required (i.e., $N \geq l(l+1)/2 + 1$) when solving \eqref{eqn: least-square problem}.

Given the set of data tuples and the knowledge of the cost function ($Q_1$, $R_{11}$, and $R_{12}$) and $H_i$, we can readily derive corresponding $\hat{Q}_1^{\pi_{i+1}} (x_k, u_{k, i}^t, v_{k, i}^t)$ and $\bar{z}_k$. 

Then, we propose an algorithm for online implementation.

\begin{algorithm}
\caption{Online Derivation of Team-optimal Solution using Q-learning-based ADP}
\label{alg: online derivation of team-optimal solution}
\begin{algorithmic}
\Require $Q_1$, $R_{11}$, $R_{12}$ (coefficient matrices of cost function), $H_0$, $x_0$, $\epsilon$;
\Ensure optimal $h = \Theta(H)$;
\State \textbf{Initialization}: $i=0$, $H_0 = 0$, $h_0 = \Theta(H_0)=0$, $P_0 = 0$, $K_{1,0} = 0$, $K_{2,0} = 0$;
\State \textbf{Step 1}: Online Data Collection
\State Collect enough data tuples ${\{x_k, x_{k+1}\}}_{k=1,2,\ldots, N-1}$, $N \geq l(l+1)/2 + 1$;
\State \textbf{Step 2}: Policy Evaluation 
\State Solve the least-square problem for $\hat{h}_{i+1}$ according to \eqref{eqn: solution of least-square problem}, and retrieve the estimate $\hat{H}_{i+1}$;
\State \textbf{Step 3}: Policy Improvement
\State Derive the new improved policy $K_{1,i+1}$ and $K_{2,i+1}$ based on \eqref{eqn: rewrite K_{1,i} using H} and \eqref{eqn: rewrite K_{2,i} using H};
\If{$\| \hat{h}_{i+1} - \hat{h}_{i} \| > \epsilon$}
    \State $i \leftarrow i + 1$, \text{go back to} \textbf{Step 2};
\ElsIf{$\| \hat{h}_{i+1} - \hat{h}_{i} \| \leq \epsilon$}
    Finish
\EndIf
\end{algorithmic}
\end{algorithm}

\subsection{Convergence to the team-optimal solution}

In this section, we will prove the effectiveness of the Algorithm \ref{alg: online derivation of team-optimal solution}, i.e., the output will converge to the optimal solution given enough samples and policy iteration numbers (sufficiently large $N$ and $i$). The convergence of the least-square problem given enough data points can be readily proved \footnote{Consider the limited space and the main focus of this work, we skip the detailed proofs. Readers may refer to \cite{nashed1974convergence, fogel1981fundamental} for details}, i.e., $\hat{h}_i \rightarrow h_i$ as $N \rightarrow \infty$. Our main focus is to prove that $h_i \rightarrow h$, $H_i \rightarrow H$, $P_i \rightarrow P$ and $Q_{1,i} \rightarrow Q_1^{\pi^t}$ as $i \rightarrow \infty$.

\begin{lemma}
    \label{lemma: convergence 1}
    The update of $h_i \rightarrow h_{i+1}$ following Algorithm \ref{alg: online derivation of team-optimal solution} is equivalent to the update of $H_i \rightarrow H_{i+1}$ defined as 
    \begin{equation}
        \label{eqn: convergence 1 lemma result}
        \begin{aligned}
           H_{i+1} &= \textbf{diag}(Q_1, R_{11}, R_{12}) + \gamma \begin{bmatrix}
            A & B_1 & B_2\\
            K_{1,i} A & K_{1,i} B_1 & K_{1,i} B_2\\
            K_{2,i} A & K_{2,i} B_1 & K_{2,i} B_2
           \end{bmatrix}^T \\
           &\cdot H_i \begin{bmatrix}
           A & B_1 & B_2\\
           K_{1,i} A & K_{1,i} B_1 & K_{1,i} B_2\\
           K_{2,i} A & K_{2,i} B_1 & K_{2,i} B_2
        \end{bmatrix}.
        \end{aligned}
    \end{equation}
\end{lemma}

\begin{proof}
    We first rewrite \eqref{eqn: recurrence on i} as
    \begin{equation}
        Q_1^{\pi_{i+1}} (z_k, \tilde{H}_i) = z_k^T \tilde{H}_i z_k, 
    \end{equation} where $z_k = {[x_k^T \ {(u_{k_i}^t)}^T \ {(v_{k_i}^t)}^T]}^T$, $u_{k_i}^t = - K_{1,i} x_k$, $v_{k_i}^t = - K_{2,i} x_k$, and
    \begin{equation}
        \begin{aligned}
            \tilde{H}_i &= \textbf{diag}(Q_1, R_{11}, R_{12}) + \gamma \begin{bmatrix}
            A & B_1 & B_2\\
            K_{1,i} A & K_{1,i} B_1 & K_{1,i} B_2\\
            K_{2,i} A & K_{2,i} B_1 & K_{2,i} B_2
            \end{bmatrix}^T \\
            &H \begin{bmatrix}
            A & B_1 & B_2\\
            K_{1,i} A & K_{1,i} B_1 & K_{1,i} B_2\\
            K_{2,i} A & K_{2,i} B_1 & K_{2,i} B_2
            \end{bmatrix}.
        \end{aligned}
    \end{equation}
    Furthermore, we substitute $Q_1^{\pi_i}(x_k, u_k, v_k) = z_k^T \tilde{H}_i z_k = {\bar{z}_k}^T \Theta(\tilde{H}_i)$ into the least-square solution as defined in \eqref{eqn: solution of least-square problem}, we have
    \begin{equation}
        h_{i+1} = {(ZZ^T)}^{-1}ZZ^T \Theta(\tilde{H}_i) = \Theta(\tilde{H}_i).
    \end{equation}

    Since $h_{i+1} = \Theta(H_{i+1})$, we have $H_{i+1} = \tilde{H}_i)$ which leads to \eqref{eqn: convergence 1 lemma result}.
\end{proof}

\begin{lemma}
    \label{lemma: convergence 2}
    The matrices $H_{i+1}$, $K_{1,i+1}$ and $K_{2,i+1}$ can be rewritten as functions of $P_i = [I \ K_{1,i}^T \ K_{2,i}^T] H_i {[I \ K_{1,i}^T \ K_{2,i}^T]}^T$ as
    \begin{equation}
        \label{eqn: lemma convergence 2 eq1}
        \begin{aligned}
        & H_{i+1}\\
        & = \begin{bmatrix}
        Q_1 + \gamma A^T P_i A & \gamma A^T P_i B_1 & \gamma A^T P_i B_2\\
        \gamma B_1^T P_i A & R_{11} + \gamma B_1^T P_i B_1 & \gamma B_1^T P_i B_2\\
        \gamma B_2^T P_i A & \gamma B_2^T P_i B_1 & R_{12} + \gamma B_2^T P_i B_2
        \end{bmatrix}
    \end{aligned}
    \end{equation}
    \begin{equation}
        \label{eqn: K_{i+1,j} iteration}
        K_{j,i+1} = \gamma{(R_{1j} + \gamma F_j B_j)}^{-1} F_j A
    \end{equation}
    \begin{equation}
        \label{eqn: F_{j} iteration}
        \begin{aligned}
             F_{j} = B_j^T P_i \big[I - \gamma B_k{[R_{1k} + \gamma B_k^T P_i B_k]}^{-1} B_k^T P_i\big],& \\
             \ j,k=1,2, \ j \neq k&
        \end{aligned}
    \end{equation}
\end{lemma}
\begin{proof}
    We can rewrite \eqref{eqn: convergence 1 lemma result} in Lemma \ref{lemma: convergence 1} as
    \begin{equation}
        \label{eqn: convergence proof of lemma 2 H_i+1}
        \begin{aligned}
            H_{i+1} =& \textbf{diag}(Q_1, R_{11}, R_{12}) + {[A \ B \ E]}^T [I \ K_{1,i}^T \ K_{2,i}^T] H_i \\
            &\cdot {[I \ K_{1,i}^T \ K_{2,i}^T]}^T [A \ B \ E]
        \end{aligned}
    \end{equation}

    The relation between $P_i$ and $H_i$ is according to \eqref{eqn: H and P relation}. Substituting \eqref{eqn: H and P relation} into \eqref{eqn: convergence proof of lemma 2 H_i+1} leads to \eqref{eqn: lemma convergence 2 eq1}. Based on \eqref{eqn: lemma convergence 2 eq1}, \eqref{eqn: rewrite K_{1,i} using H} and \eqref{eqn: rewrite K_{2,i} using H}, we derive \eqref{eqn: K_{i+1,j} iteration} and \eqref{eqn: F_{j} iteration}.
\end{proof}

\begin{lemma}
    \label{lemma: convergence 3}
    The update of $H_i \rightarrow H_{i+1}$ as \eqref{eqn: convergence 1 lemma result} is equivalent to the update of $P_i \rightarrow P_{i+1}$ as
    \begin{equation}
        \label{eqn: lemma convergence 3 P iteration}
        \begin{aligned}
            &P_{i+1} = \gamma A^T P_i A - [A^T P_i B_1 \ A^T P_i B_2] \\
            &\cdot \begin{bmatrix}
            R_{11} + \gamma B_1^T P_i B_1 & \gamma B_1^T P_i B_2 \\
            \gamma B_2^T P_i B_1 & \gamma[R_{12} + \gamma B_2^T P_i B_2] 
            \end{bmatrix}\\
            & \cdot {[A^T P_i B_1 \ A^T P_i B_2]}^T,
        \end{aligned}
    \end{equation} where $P_i = [I \ K_{1,i}^T \ K_{2,i}^T] H_i {[I \ K_{1,i}^T \ K_{2,i}^T]}^T$.
\end{lemma}
\begin{proof}
    Since $P_{i+1} = [I \ K_{1,i+1}^T \ K_{2,i+1}^T] H_{i+1} {[I \ K_{1,i+1}^T \ K_{2,i+1}^T]}^T$, we substitute $H_{i+1}$ using \eqref{eqn: lemma convergence 2 eq1} in Lemma \ref{lemma: convergence 2}, and have
    \begin{equation}
        \label{eqn: lemma convergence 3 P_{i+1}}
        \begin{aligned}
            P_{i+1} =& Q_1 + K_{1,i+1}^T R_{11} K_{1,i+1} + K_{2,i+1}^T R_{12} K_{2,i+1}\\
            &+ \gamma (A^T + K_{1,i+1}^T B_1^T + K_{2,i+1}^T B_2^T) P_i\\
            &\cdot(A + B_1 K_{1,i+1} + B_2 K_{2,i+1})
        \end{aligned} 
    \end{equation}

    Then, we substitute \eqref{eqn: K_{i+1,j} iteration}, \eqref{eqn: F_{j} iteration} into \eqref{eqn: lemma convergence 3 P_{i+1}}, which leads to \eqref{eqn: lemma convergence 3 P iteration}.
\end{proof}

Now, we are ready to state the main theorem for the convergence of policy iteration to the optimal solution.

\begin{theorem}
    \label{theorem: convergence of policy iteration for team-optimal solution}
    Consider the joint optimization problem captured by \eqref{eqn: dynamics} and \eqref{eqn: cost function of leader}.
    Given enough samples, the policy iteration process in Algorithm \ref{alg: online derivation of team-optimal solution} is equivalent to the iterating process of $H_i$ as in \eqref{eqn: convergence 1 lemma result}, and will converge to the optimal solution, i.e., $H_i \rightarrow H$, where $H$ corresponds to the optimal action-value function $Q_1^{\pi^t}$ (as in \eqref{eqn: optimal H}), and $P_i \rightarrow P$ where $P$ corresponds to the state-value function $V_1^{\pi^t}$ and solve the generalized algebraic Riccati equation (GARE).
    \begin{equation}
        \label{eqn: GARE}
        \begin{aligned}
            P =& \gamma A^T P A - [A^T P B_1 \ A^T P B_2] \\
            &\cdot \begin{bmatrix}
            R_{11} + \gamma B_1^T P B_1 & \gamma B_1^T P B_2 \\
            \gamma B_2^T P B_1 & \gamma[R_{12} + \gamma B_2^T P B_2] 
            \end{bmatrix}\\
            & \cdot {[A^T P B_1 \ A^T P B_2]}^T,
        \end{aligned}
    \end{equation}    
\end{theorem}

\begin{proof}
    In \cite{prokhorov1997adaptive}, it is shown that the iterating process of \eqref{eqn: lemma convergence 3 P_{i+1}}, starting from $P_0=0$, would converges the solution of \eqref{eqn: GARE}, i.e., $P$, corresponding to the state-value function $V_1^{\pi^t}$. Since Lemma \ref{lemma: convergence 3} proves that $H_i \rightarrow H$ is equivalent to $P_i \rightarrow P$, and $H_0$ implies that $P_0=0$ (based on \eqref{eqn: H and P relation}), we have $H_i \rightarrow H$, where $H$ corresponds to the optimal action-value function $Q_1^{\pi^t}$ (as in \eqref{eqn: optimal H}).
\end{proof}

\subsection{Follower's optimization problem}

By observing \eqref{eqn: P_v} in Lemma \ref{lemma: follower optimal strategy} and \eqref{eqn: ARE of team-optimal} in Lemma \ref{lemma: team-optimal}, we notice the similar structure of two AREs for joint optimization and follower's optimization problems, respectively. Besides, according to the closed form of incentive matrix $M$, the terms depending on the system dynamics are $\gamma A^T P_v B_1$, $\gamma B_1^T P_v B_1$, $B_2^T P_v B_1$, $A^T P_v B_2$, $B_1^T P_v B_2$, and $B_2^T P_v B_2$, which can be easily retrieved by $H_{xu}$, $H_{uu}$, $H_{vu}$, $H_{xu}$, $H_{uv}$ and $H_{vv}$ in \eqref{eqn: H} where matrix $P$ is replaced with $P_v$.

Thus, unlike solving for the previous team-optimal solution, we do not need to formulate and estimate a separate $H$ for deriving an estimate of $M$. Instead, we can directly utilize the results of the team-optimal problem including the convergence proof Theorem \ref{theorem: convergence of policy iteration for team-optimal solution} and Algorithm \ref{alg: online derivation of team-optimal solution} by replacing the leader's cost function coefficients with the follower's.

Then, Algorithm \ref{alg: online derivation of M} is proposed for deriving the estimate of $M$, denoted as $\hat{M}_i$, in an online and model-free fashion. Based on Theorem \ref{theorem: convergence of policy iteration for team-optimal solution}, it is readily to prove that $\hat{M}_i \rightarrow M$ as $N \rightarrow \infty$ and $i \rightarrow \infty$, where $M$ is the optimal solution as defined in \eqref{eqn: M}.

\begin{algorithm}
\caption{Online Derivation of Incentive Matrix using Q-learning-based ADP}
\label{alg: online derivation of M}
\begin{algorithmic}
\Require $Q_2$, $R_{21}$, $R_{22}$ (coefficient matrices of follower's cost function), $H_0$, $x_0$, $\epsilon$, $K_1$ and $K_2$ derived by using Algorithm \ref{alg: online derivation of team-optimal solution};
\Ensure optimal $M$;
\State \textbf{Initialization, and Step 1 $\sim$ Step 3}: Same as Algorithm \ref{alg: online derivation of team-optimal solution}
\If{$\| \hat{h}_{i+1} - \hat{h}_{i} \| > \epsilon$}
    \State $i \leftarrow i + 1$, \text{go back to} \textbf{Step 2};
\ElsIf{$\| \hat{h}_{i+1} - \hat{h}_{i} \| \leq \epsilon$}
    \text{move forward to} \textbf{Step 4};
\EndIf
\State \textbf{Step 4}: Reconstruction of $M$
\State Derive $M$ based on \eqref{eqn: M} and \eqref{eqn: H}
\end{algorithmic}
\end{algorithm}

\section{Conclusion}
\label{section: Conclusions}

In this paper, motivated by the physical security concerns in Industrial Control and Power Systems (ICPS), we address the incentive Stackelberg game for the resilient controller. We establish a sufficient condition for the existence of an incentive Stackelberg solution, along with the closed-form expression for the incentive matrix, given the known system dynamics. Furthermore, we introduce a Q-learning-based Adaptive Dynamic Programming (ADP) approach to determine the incentive Stackelberg solution without requiring knowledge of the system dynamics. Two algorithms are proposed to online derive the team-optimal solution and the incentive matrix, respectively. The convergence of both algorithms to the optimum solution is proven.

The major limitation of the current approach comes from the Assumption \ref{assumption: information structure}. Although Assumption \ref{assumption: information structure} seems restrictive, it is realistic in the practical physical security scenario. 
Additionally, it could serve as a foundation for exploring non-trivial extensions by incorporating alternative representations, as demonstrated in \cite{basar1979closed} and \cite{li2002approach, lin2022incentive}.

\bibliographystyle{IEEEtran}
\bibliography{main.bib}



\end{document}